\documentclass[a4paper,USenglish,cleveref,autoref,thm-restate]{lipics-v2021}

\hideLIPIcs

\usepackage{cite}
\usepackage{textcomp}

\usepackage[utf8]{inputenc}
\usepackage[T1]{fontenc}
\usepackage{amsmath,amsfonts,amsthm,amssymb,thmtools,thm-restate,booktabs,color,doi,graphicx,latexsym,url,xcolor,xspace}
\usepackage[numbers,sort&compress]{natbib}
\usepackage{microtype,hyperref}
\usepackage[inline]{enumitem}
\usepackage{soul}
\usepackage{bm}
\usepackage{thm-restate}
\usepackage{balance}
\usepackage[normalem]{ulem}
\setlist[itemize]{label=--}
\setlist[enumerate]{label=(\arabic*),labelindent=\parindent,leftmargin=*}

\usepackage{cleveref}
\crefname{figure}{Fig.}{Figs.}
\crefname{algocf}{Alg.}{Algs.}

\usepackage{tikz}
\usetikzlibrary{arrows}
\usetikzlibrary{arrows.meta}
\usetikzlibrary{positioning}
\usetikzlibrary { decorations.pathmorphing, decorations.pathreplacing, decorations.shapes, } 

\usepackage[ruled, lined, linesnumbered, commentsnumbered, noend, ]{algorithm2e}
\usepackage[noend]{algpseudocode}

\usepackage{todonotes}

\definecolor{citecolor}{HTML}{0000C0}
\definecolor{urlcolor}{HTML}{000080}

\definecolor{blueLink}{rgb}{0,0.2,0.8}
\hypersetup{colorlinks,linkcolor=blueLink,urlcolor=blueLink,citecolor=blueLink}

\newcommand{\OPT}{\textsc{OPT}}
\newcommand{\OPTstar}{\textsc{OPT}\mbox{$^*$}}
\newcommand{\ALG}{\textsc{ALG}}

\newcommand{\ONL}{\textsc{ONL}}
\newcommand{\OFF}{\textsc{OFF}}

\newcommand{\cand}{\mathcal{C} }

\newcommand{\VMs}{nodes\xspace}
\newcommand{\VM}{node\xspace}

\title{
  Online Graph Embedding in Star Graphs
}

\author{Julien Dallot}{TU Berlin, Germany}{julien.dallot@tu-berlin.de}{}{}
\author{Darya Melnyk}{TU Berlin, Germany}{melnyk@tu-berlin.de}{}{}
\author{Maciej Pacut}{TU Berlin, Germany}{maciej@inet.tu-berlin.de}{}{}
\author{Stefan Schmid}{TU Berlin and Weizenbaum Institute, Germany}{stefan.schmid@tu-berlin.de}{}{}

\authorrunning{J. Dallot, D. Melnyk, M. Pacut and S. Schmid}

\ccsdesc[500]{Theory of computation~Online algorithms}
\ccsdesc[500]{Theory of computation~Random projections and metric embeddings}
\ccsdesc[500]{Theory of computation~Scheduling algorithms}

\keywords{
  Graph Embedding, Online Algorithms, Competitive Analysis
}

\funding{
  Research funded by the German Research Foundation (DFG), Schwerpunktprogramm SPP 2378 (ReNO-2), grant 511099228, 2025-2029
}

\nolinenumbers

\begin{document}

\maketitle

\begin{abstract}
  Graph embedding is a fundamental problem of mapping nodes of a guest graph into a host graph while minimizing the distance distortion, with broad applications, including virtual network embeddings into physical topologies, VLSI design, or community detection in social networks.
  However, in many real-world applications the guest graph changes over time and the embedding can adapt to these changes (e.g. virtual machine migration in network embeddings).
  Static embeddings are inherently inefficient in comparison to adaptive embeddings, but it remains an unresolved algorithmic challenge to design efficient embedding algorithms that adapt to the demand on-the-fly, i.e., that are \emph{online}.
  
  In this paper, we derive optimal deterministic and randomized online algorithms for the online graph embedding problem in star host graphs.
  This is an essential building block on the way to design algorithms for more complex host graphs, representing a single node and its neighborhood.
  We start by presenting a $1.5$-competitive deterministic algorithm and showing that no deterministic algorithm can perform better.
  Our main contribution is a randomized algorithm that achieves a significantly better competitive ratio of $11/9 \approx 1.222$.
  Both the deterministic and the randomized algorithms are optimal, which we prove by deriving tight lower bounds for the competitiveness of any algorithm.
\end{abstract}

\section{Introduction}

Consider a graph embedding problem where the task is
to embed a guest graph $G$ into a host graph $H$ while minimizing the
total stretch of the edges~\cite{Diaz02}.
More formally, it is to find a bijective mapping $\phi: V(G) \to
V(H)$ that minimizes the quantity
\begin{align*}
  \sum_{uv \in E(G)} d_H[\phi(u), \phi(v)],
\end{align*}  
where $d_H$ is the distance in $H$ and $E(G)$ is the set of edges of $G$.

In this paper, we study the \emph{online} variant of the graph
embedding problem where the edges of the guest graph are not requested
all at once but rather one by one (an edge can be requested multiple times).
At the beginning, we are given an \emph{initial} mapping $\phi_0$ of the nodes of the guest graph into the nodes of the host graph.
Once an edge is requested, we must \emph{serve} it by paying a cost
equal to the distance between its endpoints in the current mapping.
At any time, we have the possibility to \emph{adjust} the current
mapping $\phi$, incurring a cost equal to the number of adjacent swaps
between guest nodes.
Given $\sigma = \{u_1, v_1\}, \{u_2, v_2\}, \ldots$ the sequence of
requested edges, our objective is to find a sequence of mappings
$\phi = \phi_1, \phi_2, \ldots$ to minimize the total cost for serving and adjusting:
\begin{align*}
  \min_{\phi} \sum_{i} d_{\tau}[\phi_{i-1}, \phi_i] + d_H[\phi_i(u_i), \phi_i(v_i)],
\end{align*}
where $d_{\tau}$ is the cost of adjusting the mapping and $d_H$ is the distance in the host graph.
\cref{fig:star_topology} presents a sample star topology and the costs corresponding to different operations. 

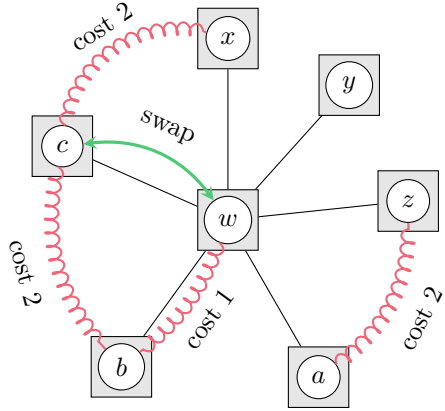
\begin{figure}[h!]
  \centering
  \begin{tikzpicture}[scale=1.2]
    \label{star_poa}
    \def \n {6}
    \def \N {8}
    \def \radius {2cm}
    \def \rd {1mm}
    \def \rer {4mm}
    \def \squarewidth {0.8cm}
    \definecolor{red1}{HTML}{EE6677}
    \definecolor{green1}{HTML}{50C878}

    \def \margin {8} 

    \node[draw, minimum size=\squarewidth, fill=black!10] (w_square) at (360:0mm) {};
    \node[draw, minimum size=\squarewidth, fill=black!10] (x_square) at ({360/\n *\n / 4}:\radius) {};
    \node[draw, minimum size=\squarewidth, fill=black!10] (y_square) at ({360/4 - 360/\n * (1.7 - 1)}:\radius) {};
    \node[draw, minimum size=\squarewidth, fill=black!10] (z_square) at ({360/4 - 360/\n * (2.4 - 1)}:\radius) {};
    \node[draw, minimum size=\squarewidth, fill=black!10] (a_square) at ({360/4 - 360/\n * (3.5 - 1)}:\radius) {};
    \node[draw, minimum size=\squarewidth, fill=black!10] (b_square) at ({360/4 - 360/\n * (4.6 - 1)}:\radius) {};
    \node[draw, minimum size=\squarewidth, fill=black!10] (c_square) at ({360/4 - 360/\n * (5.9 - 1)}:\radius) {};

    \node[draw, circle, fill=white] (w) at (360:0mm) {$w$};
    \node[draw, circle, fill=white] (x) at ({360/\n *\n / 4}:\radius) {$x$};
    \node[draw, circle, fill=white] (y) at ({360/4 - 360/\n * (1.7 - 1)}:\radius) {$y$};
    \node[draw, circle, fill=white] (z) at ({360/4 - 360/\n * (2.4 - 1)}:\radius) {$z$};
    \node[draw, circle, fill=white] (a) at ({360/4 - 360/\n * (3.5 - 1)}:\radius) {$a$};
    \node[draw, circle, fill=white] (b) at ({360/4 - 360/\n * (4.6 - 1)}:\radius) {$b$};
    \node[draw, circle, fill=white] (c) at ({360/4 - 360/\n * (5.9 - 1)}:\radius) {$c$};

    \path (w_square) edge (x_square) edge (y_square) edge (z_square) edge
    (a_square) edge (b_square) edge (c_square);

    \draw[draw=red1, bend right=27,
    decorate,decoration={coil,segment length=5pt}, thick] (a) to
    node[below=0.2, sloped] {cost 2} (z)
    ;

    \draw[draw=red1, bend right=27,
    decorate,decoration={coil,segment length=5pt}, thick] (c) to
    node[below=0.2, sloped] {cost 2} (b)
    ;

    \draw[draw=red1, bend left=15,
    decorate,decoration={coil,segment length=5pt}, thick] (w) to
    node[below=0.1, sloped] {cost 1} (b)
    ;

    \draw[draw=red1, bend left=48,
    decorate,decoration={coil,segment length=5pt}, thick] (c) to
    node[above=0.1, sloped] {cost 2} (x)
    ;

    \draw[stealth-stealth, draw=green1, bend left=30,
    decorate,decoration={segment length=8pt}, very thick] (c) to
    node[above=0.1, sloped] {{swap}} (w)
    ;

  \end{tikzpicture}

  \caption{
    A depiction of the graph embedding problem in a star host graph.
    The host graph has square, grayed nodes.
    The guest graph has circle nodes $w, x, y, \dots c$ and has red, curly edges.
    Next to each guest edge is written the cost to serve that edge in
    the depicted embedding.
    The green, straight edge shows a possible adjustment of the embedding.
    Guest edges are revealed one by one, for example $(a, z)$, $(b, w)$,
    $(c, b)$ and finally $(c, x)$.
    For that sequence of four edges, an optimal strategy first serves
    $(a, z)$ at cost $2$, $(b, w)$ at cost $1$, then swaps $c$ and $w$
    at cost $1$, and finally serves $(c, b)$  and $(c, x)$ at cost $1$ each.
  }
  \label{fig:star_topology}
\end{figure}

An inherent challenge in many graph embedding applications  is that
the edges of the guest graph are not known in advance but rather
requested one by one (notably in virtual network embedding, see
\cref{sec:motivation}).
A well-established method to deal with uncertainty of the future is
the framework of online algorithms and competitive
analysis~\cite{Borodin1998}.
Ideally, these online algorithms achieve a good competitive ratio:
intuitively, without knowing the future demand, their performance is
almost as good as an optimal offline algorithm that knows the future.
Formally, for a given request sequence $\sigma$, let $\ONL(\sigma)$ and
$\OPT(\sigma)$ be the cost incurred by an online algorithm $\ONL$ and
an optimal offline algorithm $\OPT$, respectively.
$\ONL$ is said to be $c$-competitive if for any input
sequence $\sigma$ it holds that $\ONL(\sigma) \le c\cdot
\OPT(\sigma)$.

Online graph embedding has been attracting attention in recent years.
The community is currently in the process of designing efficient algorithms for general host topologies, starting with simple topologies due to the challenging nature of the problem: capacitated cliques \cite{Henzinger0019, PacutP020, HenzingerNRS21} or lines~\cite{ILU, BienkowskiE24, LearningMLA}.
We complement this line of research by studying the online graph embedding in the star graph, which addresses fundamentally different aspects of the problem and a crucial building block for more complex host topologies.

Our problem is also related to classic online problems in data structures, such as the \emph{list access problem} and self-adjusting binary search trees.
In these problems, requests typically pertain to a \emph{single node}, and the goal is to minimize access and reorganization costs~\cite{SleatorT85,Reingold1994}.
In contrast, the considered problem---and similar problems like dynamic graph partitioning~\cite{AvinBLPS20,Henzinger0019} and dynamic minimum linear arrangement~\cite{ILU}---features requests between \emph{pairs of nodes}.
This adds a significant algorithmic challenge, as an optimal solution must consider not only the positioning of individual nodes but also how their pairwise interactions evolve over time.

\subsection{Our contributions}

We develop optimal deterministic and randomized online algorithms for the online graph embedding problem in star host graphs.
Our algorithms operate without any assumptions about the input distribution.
Despite that, we derive formal guarantees on their performance, comparing with a hypothetical offline algorithm that knows the future.
We complement our upper bounds with tight lower bounds, showing that our algorithms are the best possible.

\medskip
\noindent\textbf{Deterministic algorithms.}
We show a lower bound of $1.5$ for the competitive ratio of any deterministic algorithm for the online graph embedding problem.
This bound is derived on a star topology with only three physical nodes.
We match this lower bound by deriving an optimal $1.5$-competitive deterministic algorithm.
This algorithm keeps track of the guest graph nodes that the optimal offline algorithm might have at its center (the \emph{candidates set}).
When the algorithm narrows down the set of such \VMs\ to at most two, it migrates any of them to the center (cf. \cref{sec:det}).

\medskip
\noindent\textbf{Randomized algorithms.} 
The main contribution of this paper is a $11/9$-competitive randomized algorithm.
This algorithm keeps the general outline given by the deterministic algorithm, and improves it in two aspects: (1) when the randomized algorithm narrows down the set of candidates to two nodes, it breaks ties randomly, and (2) while the set is not yet narrowed down to two, it still moves the requested node to the center with some probability. 
We also derive a lower bound of $11/9$ showing that the randomized upper bound is tight.
To this end, we present an input distribution based on two request patterns for which any deterministic algorithm is at least $11/9$-competitive and use Yao's principle to derive the lower bound (cf. \cref{sec:rand-d=1}). 

\subsection{Motivation: Online virtual network embedding}
\label{sec:motivation}

With the popularity of data-centric workloads, e.g., related to batch processing or distributed machine learning, the efficiency of the interconnecting network has become critical for the performance of many applications.
The requirements on the network are further increased by trends toward resource disaggregation (where fast access to remote resources e.g., GPUs or memory, is critical) and toward hardware-driven workloads (such as distributed training)~\cite{li2019hpcc,khani2021sip,mogul2012we}.
Accordingly, over the last years, great efforts have been made to improve the performance of communication networks, especially in datacenters~\cite{pieee19}.

A particularly attractive approach to rendering communication networks more efficient is to render them traffic-aware and ``self-adjusting''.
This can be done by leveraging resource allocation flexibilities enabled by virtualization:  frequently communicating nodes (e.g.,~containers or virtual machines) can be migrated topologically closer, thus saving resources and reducing communication overheads (including ``bandwidth taxes''~\cite{MelletteMRFPSP17,GrinerZBG0A22}).
Empirical studies indeed show that communication traffic features much spatial and temporal structure~\cite{kandula2009nature,BensonAM10}, which may be  exploited for optimization.

The algorithmic challenge of such optimizations is to strike a balance between the benefits and the costs of node migrations.
Ideally, a node migration should be amortized by a more efficient resource allocation in the future, and frequent node migrations should be avoided, e.g., due to instabilities.

This algorithmic problem features interesting connections to classic
optimization problems.
In particular, the problem can be seen as a dynamic version of the virtual network embedding problem~\cite{Rost019}, which asks for an allocation of a virtual network on a physical network that keeps frequently communicating nodes topologically close together.
In our dynamic version of this problem, this allocation can be adjusted.

With the self-adjusting perspective on the embedding problem, we are in the realm of competitive analysis and are interested in online algorithms which are competitive against an optimal offline algorithm.
Competitiveness is the price for decision-making without knowledge of the future, and strong competitiveness guarantees may yield sufficient bounds for the performance of the networked system~\cite{Borodin1998}, e.g., without the need to design complex traffic prediction algorithms.

\subsection{Related work}
\label{sec:relwork}

The dynamic graph embedding problem is related to several well-known online problems in the field of data structures, where requests typically pertain to a single node. For instance, in the classic \emph{list access problem}~\cite{SleatorT85} and in self-adjusting binary search trees~\cite{SleatorT85tree}, each request involves accessing a single node. In these problems, the goal is to minimize access and reorganization costs by optimizing the placement of individual nodes. Similar concepts are also applied in the context of virtual machine migration in computer networks, where the focus is typically on individual nodes~\cite{Reingold1994,KamaliSurvey2013}. 

In contrast, our work, along with problems like dynamic graph partitioning~\cite{AvinBLPS20,Henzinger0019,PacutP020,PacutP021,HenzingerNRS21} and dynamic minimum linear arrangement~\cite{ILU}, introduces a significantly different challenge: the requests in these problems involve \emph{pairs of nodes}, not just single nodes. This fundamental difference adds complexity to the algorithmic design, as the optimization must account for the interactions between pairs of nodes over time. The embedding must be adjusted dynamically to minimize the communication cost between frequently interacting nodes, which requires a more sophisticated approach than the single-node cases. This makes the problem of dynamic graph embedding in star topologies particularly challenging, where a central node may have a reduced communication cost with all other nodes, and decisions must be made not only for individual nodes but for the pairwise communication patterns between them. 

For example, in the graph partitioning problem, a network of clusters of virtual machines is considered, where requests between nodes in the same cluster cost $0$ and those across clusters cost $1$~\cite{AvinBLPS20,Henzinger0019}. Similarly, in the dynamic minimum linear arrangement problem~\cite{ILU}, the challenge arises from requests between pairs of nodes, with the objective of minimizing the cost of communication based on the distance between the nodes.

Thus, while the online problems in data structures focus on optimizing access to individual nodes, our problem, like dynamic graph partitioning and dynamic minimum linear arrangement, involves pairwise node interactions, requiring novel algorithmic approaches to handle the added complexity of managing multiple interacting nodes over time.

\section{Model and definitions}

\label{sec:model}
The dynamic graph embedding problem in star topologies is formulated as follows.
We are given a host graph $N$ consisting of $n$ physical hosts connected in a \emph{star topology}.
In this topology, a distinguished \emph{central} host is at distance $1$ from every other host.
Every other pair of hosts is at distance $2$ from each other. 
Hence, we use the same cost model as in dynamic minimum linear arrangement from Olver et al.~\cite{ILU}.

The $n$ nodes of the guest graph $x\in X$ are embedded in the nodes of the host graph $N$.
A mapping from $X$ to the nodes of the host graph $N$ is called a \emph{configuration}.
In a configuration, at most one \VM\ can be mapped to a physical host. As the number of hosts is equal to the number of \VMs, each host has exactly one \VM mapped to it.

The requests are issued in a sequence $\sigma = (\sigma_1, \sigma_2,$ $\sigma_3, \ldots)$.
Each request $\sigma_i\in \sigma$ consists of a pair of nodes of the guest graph, i.e., $\sigma_i = \{x_k,x_l\}$ where $x_k, x_l\in X$.
After a request is issued, we must \emph{serve} it by paying a cost equal to the distance between the respective \VMs\ in the current configuration.

At any time (also before serving a request), an algorithm may alter the assignment of the guest nodes to the host nodes by choosing a new center, at the cost of $1$.
Note that a reasonable algorithm would always choose to migrate a~\VM\ to the center after seeing and before serving a request.
Our goal is to find an algorithm that minimizes the total cost over the entire sequence $\sigma$.
In addition, this algorithm should make its decisions in the \emph{online} framework, i.e., without knowledge of future requests.

\subsection{Online algorithms and competitive analysis}
\label{apx:competitive}

In the online version of the dynamic graph embedding problem, only one request from $\sigma$ is presented to the online algorithm \ALG\ at a time.
As in the offline case, the goal of the online algorithm is to minimize the total serving and migration costs for any given request sequence. 

We use competitive analysis to determine the quality of the online algorithm.
We thereby compare the online algorithm to the best possible offline algorithm \OPT\ that has full knowledge of the request sequence.
In order to have a fair comparison, we assume that both algorithms start in the same configuration.
We define $\ALG(\sigma)$ and $\OPT(\sigma)$ to be the costs of the online and the offline algorithm on the input sequence $\sigma$ respectively.
In the deterministic case, an algorithm \ALG\ is called
$c$-competitive if, for every possible input sequence $\sigma$, it
satisfies the inequality $\ALG(\sigma) \le c\cdot \OPT(\sigma)$.

When considering the randomized case, we assume that the adversary that presents the input sequence, is oblivious.
That is, the adversary has to determine the entire input sequence in
advance, without the knowledge of the actual random choices of \ALG\
(but only its probability distribution)~\cite{Borodin1998}.
In the randomized case, an algorithm \ALG\ is called $c$-competitive,
for every possible input sequence $\sigma$, the inequality
$\mathbb{E}(\ALG(\sigma)) \le c\cdot \OPT(\sigma)$ is satisfied.

\section{Deterministic algorithms}\label{sec:det}

We present a deterministic, online algorithm for the dynamic graph embedding problem in the star topologies.
We then prove that this algorithm is $1.5$-competitive.
Finally, we show that this competitive ratio is optimal by deriving a lower bound of $1.5$ for deterministic algorithms.

\subsection{A deterministic online algorithm}

In this section, we present the deterministic online algorithm PivotTracking that solves the dynamic graph embedding problem in star topologies in the deterministic framework. We refer to this algorithm as \ALG\ in the rest of the section.
We further introduce \OPTstar, an offline algorithm with optimal cost.
Among all optimal offline algorithms, \OPTstar\ is carefully chosen so that \ALG\ can track \OPTstar's central \VM; this tracking property is described in \cref{lem:tracking}. 
Finally, we prove that \ALG\ is $1.5$-competitive using this tracking property.
Our lower bound above shows that this upper bound is optimal.

The main design idea is to maintain a set of \VMs\ that an optimal offline algorithm \OPTstar\ potentially placed at the center.
In the rest of this article, this set is called the \emph{candidate set}, it is denoted by $\cand$.
\ALG\ updates $\cand$ every time a new request is issued; it can update $\cand$ in two ways, referred to as the union and the intersection behavior.
When a new request is issued, \ALG\ may either:
\begin{enumerate}
\item Confirm that \OPTstar\ currently has a \VM\ from the candidate set at the center.
  Then the algorithm would apply the \emph{intersection behavior} to reduce the candidate set to at most $2$ \VMs\ and place one of them at the center.
\item Discover a new candidate \VM\ (at most $2$ at once).
  In this case, we apply the \emph{union behavior} and add the new candidate(s) to $\cand$.
\end{enumerate}
\cref{alg:det_d_unit_leadertracking} presents this idea in pseudocode.

\SetKwFor{When}{when}{do}{endwhen}
\begin{algorithm}[h!]

  \caption{Deterministic PivotTracking}\label{alg:det_d_unit_leadertracking}

  \SetKwInOut{Input}{Input}
  \Input{let $x_{\textrm{init}}$ be the initial \VM\ at the center, $\cand \gets \{x_{\textrm{init}}\}$}

  \When{request $\sigma_t =\left\{x_{1}, x_{2}\right\}$ is issued}
  {
    \tcp{intersection behavior}
    \uIf{$|\cand \cap \sigma_{t}| > 0$}
    { 
      $\cand \gets \cand \cap \sigma_{t}$\\
      \uIf{no \VM\ from $\cand$ is at the center}{
        migrate a \VM\ from $\cand$ to the center, breaking ties according to a fixed order
      }
    }
    \tcp{union behavior}
    \uElse{
      $\cand \gets \cand \cup \sigma_{t}$\\
    }
    serve request $\sigma_t$
  }
\end{algorithm}
\newcommand{\detCompRatio}{1 + \frac{r}{r+2}}

In order to show that \cref{alg:det_d_unit_leadertracking} is
$1.5$-competitive, we will compare its performance to a particular
optimal offline algorithm \OPTstar\ that we define below.
To define \OPTstar, we first introduce the concept of phases defined
for a given feasible algorithm.

\begin{definition}[A phase and a pivot]\label{def:phase_and_pivot}
  Let $\ALG$ be any algorithm solving our problem.
  A \emph{phase} of $\ALG$ is a subset of consecutive requests where
  $\ALG$ keeps the same central node.
  At any time $t$, the \emph{pivot} of $\ALG$ is the central node of
  $\ALG$ at time $t$.
\end{definition}

\begin{definition}[$\OPTstar$]\label{def:OPTstar}
  We define $\OPTstar$ as a specific optimal offline algorithm that minimizes the lexicographic order on the reversed sequence of its phase lengths.
\end{definition}

\OPTstar\ is obtained by restricting the set of all optimal offline
algorithms to the ones whose last phase has minimum length; we further
restrict these algorithms to the ones whose penultimate phase has
minimum length, and so on until the first phase.
Possible ties are broken by taking an arbitrary algorithm among the
remaining ones.

\begin{lemma}\label{lem:phaseStartsWithX}
  Let $P = \langle \sigma_{i}, \sigma_{i+1}, \dots \sigma_{j} \rangle$ be a phase of \OPTstar.
  Assume that $P$ is not the first phase of $\sigma$, and let $x$ be the pivot of \OPTstar\ during $P$.
  Then $x \in \sigma_{i}$. 
\end{lemma}
\begin{proof}[Proof idea]
  Observe that, for any phase length, a pivot \VM\ $x$ needs to be contained in one of the requests of the phase otherwise \OPTstar\ would not be optimal since there would be no compensation for the initial migration cost.
  It is possible to show that $x$ is contained in the first request by showing the inverse: if $x$ is not contained in the first request of a phase, then it is possible start the phase later while incurring the same cost. This contradicts the assumption that the algorithm has the lowest lexicographic order.\phantom\qedhere
\end{proof}

\begin{restatable}{lemma}{trackinglemma}\label{lem:tracking}
  Let $\sigma_i$ be a request in $\sigma$ and let \ALG\ denote \cref{alg:det_d_unit_leadertracking}. Then
  \begin{itemize}
  \item \ALG\ applies its intersection behavior on $\sigma_i$ $\iff$ \OPTstar\ pays $1$ for $\sigma_{i}$.
  \item \ALG\ applies its union behavior on $\sigma_i$ $\iff$ \OPTstar\ pays $2$ for $\sigma_{i}$, either by directly paying $2$ or by first migrating a \VM\ and then paying $1$.
  \end{itemize}
\end{restatable}
Due to its length, the proof of \cref{lem:tracking} is deferred to Appendix \ref{apx:tracking}.

\begin{definition}[Cheap and expensive requests]\label{def:expensive_cheap}
  Let $\sigma_{i}$ be a request of $\sigma$.
  We say that $\sigma_{i}$ is \emph{cheap} if \OPTstar\ pays $1$ for it.
  We say that $\sigma_{i}$ is \emph{expensive} if \OPTstar\ pays $2$ for it, either by directly paying $2$ or by first migrating a \VM\ and then paying $1$.
\end{definition}

In the previous lemma, we showed the connection between the algorithm structure and \OPTstar.
In order to prove the competitive ratio of the algorithm, we will arrange the input sequence with cheap and expensive requests in blocks.
We then derive the tight competitive ratio on each block in \cref{th:det-ub}.

\begin{definition}\label{def:block_decomposition}
  Let $\sigma$ be a finite request sequence.
  We use $\Gamma$ to denote a cheap request, and $\mathcal{E}$ to denote an expensive request in $\sigma$.
  We define a \emph{block} as a subsequence of successive requests of $\sigma$ containing at least one expensive request and then at least one cheap request, in this order.
  A block $B$ hence matches the following regular expression:
  \begin{align*}
    B = \mathcal{E}^{+} \Gamma^{+}
  \end{align*}
  We define a \emph{block decomposition} of $\sigma$ by decomposing $\sigma$ into the following regular expression:
  \begin{align*}
    \sigma = \Gamma^{*} \ \left(\mathcal{E}^{+} \ \Gamma^{+}\right)^{*} \mathcal{E}^{*}
  \end{align*}
\end{definition}

It can be easily verified that such a block decomposition is always possible and unique.

\begin{theorem}\label{th:det-ub}
  \cref{alg:det_d_unit_leadertracking} is $1.5$-competitive.
\end{theorem}
\begin{proof}
  Fix any sequence of requests $\sigma$. 
  To analyze the competitive ratio, we partition the block decomposition of $\sigma$ into three parts: the \emph{prefix} $\Gamma^{*}$, the sequence of \emph{blocks} $\mathcal{E}^{+} \ \Gamma^{+}$ and the \emph{suffix} $\mathcal{E}^{*}$. 
  Observe that for the prefix and the suffix, \ALG\ never pays more than \OPTstar, we therefore only need to focus on the blocks $\mathcal{E}^{+} \ \Gamma^{+}$ to analyze the competitive ratio. 

  Let $B$ be a block, we denote the number of expensive and cheap requests in $B$ by $\epsilon$ and $\gamma$, respectively.
  Let $\OPTstar(B)$ be the cost of \OPTstar\ on $B$.
  By definition, the following equation holds:
  \begin{align*}
    \OPTstar(B) = 2 \epsilon + \gamma
  \end{align*}

  As stated in \cref{lem:tracking}, \ALG's intersection behavior triggers when the $\gamma$ cheap requests of the block $B$ are presented.
  Since \ALG\ serves $\gamma$ consecutive cheap requests, $\cand$ experiences $\gamma$ consecutive intersections.
  If \ALG\ pays an expensive cost on such cheap requests, the size of $\cand$ must strictly decrease, which can happen at most two times.
  Therefore, \ALG\ pays no more than $\gamma + \min \{2, \gamma\}$ for the $\gamma$ cheap requests.
  The following upper bound for the competitive ratio holds:

  \begin{align*}
    \frac{\ALG(B)}{\OPTstar(B)} &\le \frac{2 \epsilon + \gamma + \min \{2, \gamma\}}{2 \epsilon + \gamma}\\
                                &\le 1 + \frac{\min \{2, \gamma\}}{2 \epsilon + \min \{2, \gamma\}}\\
                                &\le 1.5 \qedhere
  \end{align*}
\end{proof}

\subsection{A lower bound of 1.5 for deterministic algorithms}

\begin{theorem}\label{lem:any_r_lb-det}
  No deterministic online algorithm can achieve a better competitive ratio than $1.5$ for the online dynamic graph embedding problem in the star topology.
\end{theorem}
\begin{proof}
  Let \ALG\ denote any deterministic algorithm.
  Fix three \VMs\ $x_{1}$, $x_{2}$ and $x_{3}$.
  We will construct the sequence $\sigma$ that always requests two \VMs\ from $x_{1}, x_{2}, x_{3}$ which are not at the center in \ALG's configuration.
  Hence, \ALG\ pays at least $2$ for each request of $\sigma$.

  Given $i \in \{1,2,3\}$, we call $|\sigma(x_{i})|$ the number of requests in $\sigma$ where \VM\ $x_{i}$ is requested.
  Since each request of $\sigma$ contains exactly two different \VMs, we have $\sum_{i=1}^{3} |\sigma(x_{i})| = 2 \cdot |\sigma|$.
  Assuming that $x_{1}$ was the most requested \VM\ in $\sigma$, we have $|\sigma (x_{1})| \geq \frac{2}{3} \cdot |\sigma|$.

  Let \OFF\ be the offline algorithm that migrates the overall most requested \VM\ of $\sigma$ to the center before the first request, and does not change its configuration afterward.
  As shown by the above inequality, \OFF\ serves more than $\frac{2}{3} \cdot |\sigma|$ requests at a cost of  $1$, and it serves the remaining requests at a cost of $2$.
  Hence \OFF\ pays at most $1 \cdot |\sigma(x_{1})| + 2 \cdot (|\sigma| - |\sigma(x_{1})|) \leq \frac{4}{3} \cdot |\sigma|$ throughout the whole input sequence.
  Note that we are not accounting for the initial migration cost to migrate $x_{1}$ to the center since it becomes negligible as the number of requests grows.
  For any optimal offline algorithm \OPT, we finally have
  \[
    \frac{\ALG (\sigma)}{\OPT (\sigma)} \geq \frac{\ALG (\sigma)}{\OFF (\sigma)}  \geq \frac{2 \cdot |\sigma|}{ 4/3 \cdot |\sigma|} = 1.5
  \]
  from which the claim follows.
\end{proof}

\section{Randomized algorithms}\label{sec:rand-d=1}

The main contribution of this paper is an optimal randomized algorithm for the dynamic graph embedding problem in star topologies.
We begin this section by introducing a lower bound then we design and analyze a randomized online algorithm that matches this bound.

\subsection{A lower bound of 11/9 for randomized algorithms}\label{sec:rand_lb_small_r}

In this section, we present a lower bound of $11/9$ for the randomized \VM\ migration problem.
We first design a probabilistic input sequence $\sigma$ and analyze the expected costs that an offline algorithm can achieve on this sequence.
Using Yao's principle, we obtain a lower bound on the expected cost of any randomized online algorithm.

We build the probabilistic input sequence $\sigma$ based on successive pairs of requests.
Every pair of $\sigma$ will have an associated \VM\ called the \emph{pivot} that will be important to define further algorithmic behavior --- we call it this way as we will compare any online algorithm to an offline algorithm \OFF\ that keeps the pivot at center.
Then, we construct pairs of requests one after another according to the following scheme.
Let $p \in [0, 1]$ be a probability parameter and let $i \ge 3$ be an odd integer.
The pair of requests $(\sigma_{i}, \sigma_{i+1})$ is built according to one of the following patterns:

\begin{figure*}[h!]
  \centering
  \includegraphics[width=0.7\linewidth]{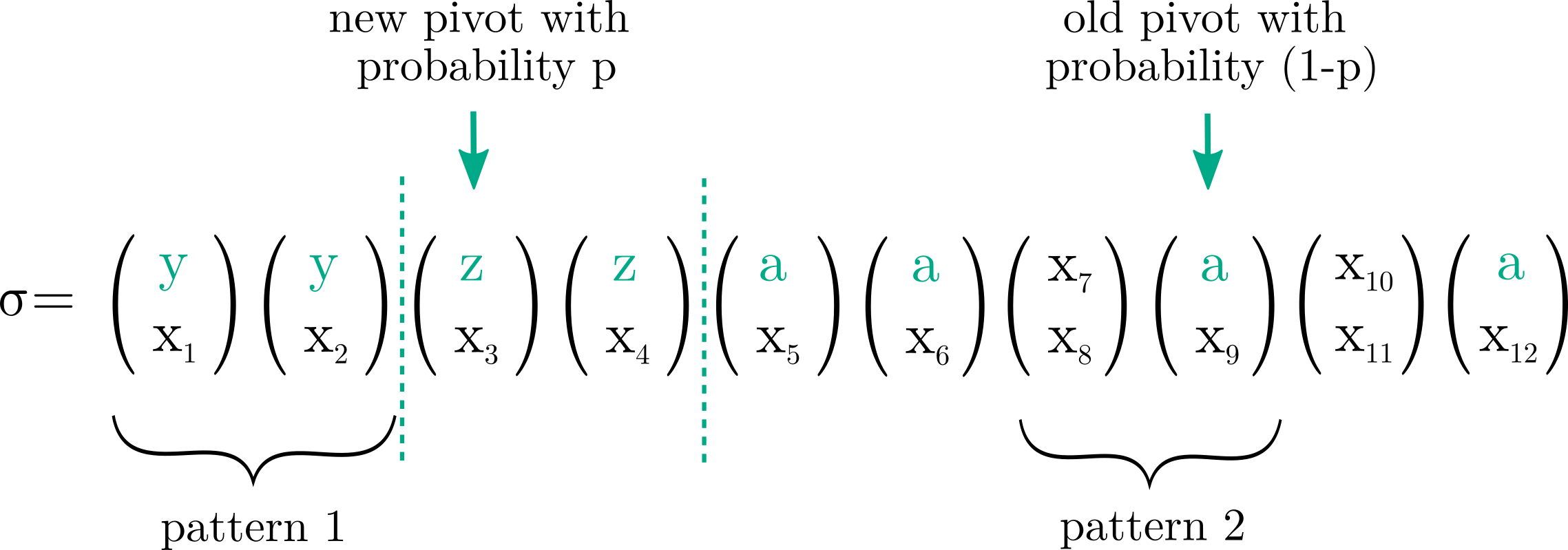}
  \caption{
    Presentation of a possible request sequence, with three patterns of type 1 followed by two patterns of type 2.
    $y, z,$ and $a$ are the successive pivots of the offline algorithm \OFF\ for this input sequence.
    The dotted lines visualize the changes of the pivot.
    $y,z,a$ and $x_i, i\in{1,\dots,12}$ are chosen uniformly at random such that they do not overlap with the nodes of the previous patterns.
    Let an offline algorithm \OFF\ start with $x_1$ at the center.
    Then \OFF\ pays $3$ for each pair of requests of this input sequence, either because it migrates a node (pattern 1) or because of a request that does not contain the pivot (pattern 2).
  }
  \label{fig:lower_bound_pair_definition}
\end{figure*}

\textbf{Pattern 1. }
Happens with probability $p$.
Pattern $1$ is defined as $(\sigma_i,\sigma_{i+1}) = \bigl((a, x_1), (a, x_2)\bigr)$, where $a \neq x_{1} \neq x_{2}$ are chosen uniformly at random among \VMs\ that did not appear in the previous pair $(\sigma_{i-2}, \sigma_{i-1})$.
The \VM\ $a$ is the pivot of this pair.

\textbf{Pattern 2.}
Happens with probability $1-p$.
Pattern $2$ is defined as $(\sigma_i,\sigma_{i+1}) = \bigl((x_1, x_2), (a, x_3)\bigr)$, where $x_{1} \neq x_{2} \neq x_{3}$ are chosen uniformly at random among \VMs\ that did not appear in the previous pair $(\sigma_{i-2}, \sigma_{i-1})$.
Here, the \VM\ $a$ is the pivot of the previous and the current pair.

The request sequence $\sigma$ starts with pattern $1$.
Afterward, pattern $1$ follows with probability $p$ and pattern $2$ with probability $1-p$.
This process is repeated indefinitely, hence $\sigma$ has infinite length.
\cref{fig:lower_bound_pair_definition} features an example for the first five pairs of $\sigma$.

In the following, we present the lower bound.
This bound holds against the oblivious adversary, the standard adversary studied in randomized online algorithms~\cite{Borodin1998}.

\begin{theorem}\label{th:lb-rand}
  No randomized online algorithm can achieve a better competitive ratio than $11/9$ for the online dynamic graph embedding problem in the star topology with $n \ge 10$ physical hosts.
\end{theorem}

\begin{proof}
  To prove this theorem, we will assume that the input sequence is chosen according to the description above.
  In particular, pattern 1 is chosen with probability $p$ and pattern 2 is chosen with probability $1-p$.
  We will determine the value of $p$ later in this proof. 

  Let \OFF\ be the following offline algorithm: just before \OFF\ serves the first request of a~pair, it migrates the pivot of that pair to the center --- it does nothing if the pivot is already at the center.
  Observe that \OFF\ pays exactly $3$ for each pair in~$\sigma$.
  
  We use \OFF\ as a reference to compute a lower bound on the expected cost of any online algorithm.
  Let \ALG\ be a deterministic online algorithm.
  Let $i \ge 3$ be an odd integer, we assume that $(\sigma_{i}, \sigma_{i+1})$ is a pair of requests that \ALG\ is about to serve.
  Let $x$ be \ALG's central \VM\ after it finished serving the previous pair.
  We now distinguish between three cases depending on \ALG's central \VM: (1) $x$ is the pivot of the previous pair, or (2) $x$ was requested in the previous pair but is not the pivot, or (3) $x$ was not requested in the previous pair.
  \begin{enumerate}
  \item \ALG's central \VM\ $x$ was the pivot of the previous pair.
    Whatever the pattern of the current pair is, \ALG\ does not have a \VM\ requested in $\sigma_{i}$ at the center and cannot know the type of the pattern yet.
    If \ALG\ chooses to migrate one of the requested \VMs\ of $\sigma_{i}$ to the center, it pays $2$ for $\sigma_{i}$ and $1.5 \cdot p + 2 \cdot (1-p)$ in expectation for $\sigma_{i+1}$.
    Otherwise, if \ALG\ chooses to not migrate a \VM\ before it serves $\sigma_{i}$, then it pays $2$ for $\sigma_{i}$ and $2 \cdot p + 1 \cdot (1-p)$ in expectation for $\sigma_{i+1}$.
    We now fix $p= 2/3$ so that the expected cost of \ALG\ on this pair is the same independent of its choice.
    The expected cost for this pair is $11/3$.

  \item \ALG's central \VM\ $x$ was requested in the previous pair but is not the pivot.
    If \ALG\ chooses to migrate a \VM\ requested in $\sigma_{i}$ before serving it, it pays $2 + 1.5 \cdot p + 2 \cdot (1-p) = 11/3$ in expectation since we fixed $p=2/3$.
    Otherwise, if \ALG\ chooses to not migrate a \VM\ before it serves $\sigma_{i}$ then \ALG\ pays $4$ for the pair.
    The expected cost of \ALG\ on this pair is at least $11/3$ independent of its choice.
  \item \ALG's central \VM\ $x$ was not requested in the previous pair.
    Then \ALG\ necessarily paid more than $4$ for the previous pair.
    Similarly to the previous cases, \ALG\ pays $11/3$ for the pair in expectation if it chooses to migrate a \VM\ that is requested in $\sigma_{i}$ before serving it.
    Otherwise, if \ALG\ chooses to not migrate a \VM\ before it serves $\sigma_{i}$, we compute the probabilities for it to pay $2$, $3$ and $4$ for this pair.
    \ALG\ pays $2$ iff the pair is of pattern $1$ and \ALG\ started the pair with the pivot at the center.
    This happens with probability $p \cdot \frac{1}{n-4}$.
    \ALG\ pays $4$ iff $x \notin \sigma_{i} \cup \sigma_{i+1}$, which happens with probability at least $\frac{n-7}{n-4}$.
    \ALG\ pays $3$ for the remaining cases.
    One can check that the expected cost on the current and previous pairs is not smaller than $2 \cdot 11/3$ for $n \ge 10$, hence the expected cost per pair is at least $11/3$.
  \end{enumerate}
  
  For any pair, except the first one, we proved that \ALG\ pays an expected cost no smaller than $11/3$ in all cases.
  The role of the first pair becomes negligible as the number of requests grows.
  Hence, any deterministic algorithm shows an expected competitive ratio no smaller than $11/9$.
  According to Yao's principle, we deduce that no randomized algorithm has a~competitive ratio smaller than $11/9$.
\end{proof}

\subsection{An optimal randomized online algorithm}

In this section, we extend the algorithm design outline for the deterministic dynamic graph embedding to the randomized setting. We thereby match the presented lower bound from the previous section. 
As \cref{alg:det_d_unit_leadertracking}, the randomized algorithm tracks the set of candidate \VMs\ that a particular offline algorithm \OPTstar\ (defined above in \cref{def:OPTstar}) may keep at the center. In the intersection behavior, the randomized algorithm breaks ties uniformly at random. In the union behavior, with some probability, the randomized algorithm has an option not to migrate any \VMs.
\cref{alg:rand_leadertracking} presents this idea as pseudocode.

\begin{algorithm}[h!]

  \caption{Randomized PivotTracking}\label{alg:rand_leadertracking}

  \SetKwInOut{Input}{Input}
  \Input{let  $x_{\textrm{init}}$ be the initial \VM\ at the center, $\cand \gets \{x_{\textrm{init}}\}$}

  \When{request $\sigma_t =\left\{x_{1}, x_{2}\right\}$ is issued}
  {
    \tcp{intersection behavior}
    \uIf{$|\cand \cap \sigma_{t}| > 0$}
    { 
      $\cand \gets \cand \cap \sigma_{t}$\\
      \uIf{no \VM\ from $\cand$ is at the center}{
        migrate a \VM\ from $\cand$ to the center, breaking ties uniformly at random
      }
    }
    \tcp{union behavior}
    \uElse{
      $\cand \gets \cand \cup \sigma_{t}$\\
      do one of the following actions with uniform probability: do nothing, migrate $x_{1}$ to the center, migrate $x_{2}$ to the center.
    }
    serve request $\sigma_t$
  }
\end{algorithm}

In the following theorem, we claim that \cref{alg:rand_leadertracking} is $11/9$-competitive. By combining this result with the lower bound from \cref{th:lb-rand} we show that this algorithm is optimal.

\begin{restatable}{theorem}{randomizedCompetitiveRatio}
  \label{th:rand-ub}
  \cref{alg:rand_leadertracking} is $11/9 \approx 1.222$ competitive against the oblivious adversary.
\end{restatable}

\begin{proof}
  Fix any sequence of requests $\sigma$, and let \OPTstar\ be an optimal algorithm that minimizes the reversed lexicographic order on $\sigma$ as described in \cref{def:OPTstar}.
  Throughout this proof, we refer to \cref{alg:rand_leadertracking} as \ALG.
  We perform a block decomposition of $\sigma$ as described in \cref{def:block_decomposition}:
  \begin{align*}
    \sigma = \Gamma^{*} \ \left(\mathcal{E}^{+} \ \Gamma^{+}\right)^{*} \mathcal{E}^{*}
  \end{align*}

  \noindent As in the proof for \cref{th:det-ub}, we focus on the blocks of the form $\mathcal{E}^{+} \ \Gamma^{+}$ in the following analysis. 

  This time, we however distinguish between two kinds of blocks.
  We say that a block is \textit{simple} if $\cand$ contains only one \VM\ when \ALG\ has served its last request.
  Conversely, a~block is called \textit{ambiguous} if $\cand$ contains two \VMs\ when the last request is served.
  Using these definitions, we further decompose the blocks of $\sigma$ into the following regular expression, where $A$ and $S$ stand for a simple and an ambiguous block respectively:
  \begin{align*}
    \sigma = \Gamma^{*} \ (A^{*} \ S)^{*} \ A^{*} \ \mathcal{E}^{*}
  \end{align*}

  Let $B_{1}, B_{2} \ldots B_{q}$ be the blocks of a pattern $A^{*} \ S$ in the above regular expression.
  Let $\epsilon_{i}$ and $\gamma_{i}$ be the numbers of expensive and cheap requests in $B_{i}$, $\forall i \in [1, q]$.
  We analyze this pattern in two steps: either $q=1$ or $q > 1$.
  We will then analyze the remaining blocks of the form $A^{*}$ to complete the proof.

  \textbf{Case $\bm{q = 1}$.}
  Let $x$ be \OPTstar's central \VM\ on the cheap requests of $B_{1}$.
  Compared to \OPTstar, \ALG\ pays an extra cost for the cheap requests unless \ALG\ already migrated $x$ to the center when the first cheap request is issued.

  In the following paragraph, we prove that, with a probability of at least $1/3^{\epsilon_{i}}$, \ALG\ has $x$ at the center when the first cheap request is issued.
  By \cref{lem:phaseStartsWithX}, $x$ was requested once before the cheap requests of $B_{1}$.
  If $x$ is requested among the expensive requests of $B_{1}$ then \ALG\ migrates $x$ to the center with probability $1/3$ since \ALG\ applies its union behavior on the expensive requests (\cref{lem:tracking}).
  If $x$ was migrated to the center, $x$ stays at the center until the first cheap request is issued with a probability of at least $1/3^{\epsilon_{1}-1}$.

  On the other hand, if $x$ was requested before $B_{1}$, we make use of the initial assumption that $q = 1$.
  It implies that the candidate set contains one \VM\ before $B_{1}$.
  Using \cref{lem:tracking}, \OPTstar\ and \ALG\ have the same central \VM\ before $B_{1}$ and this central \VM\ is $x$ ---~otherwise \OPTstar\ would start a new phase during $B_{1}$ and $x$ would have been requested in the expensive requests of $B_{1}$ using \cref{lem:phaseStartsWithX}.
  $x$ is then kept at the center until the first cheap request of $B_{1}$ with a probability of at least $1/3^{\epsilon_{1}}$.

  To prove the claimed upper bound, we make a distinction between two cases: either $\gamma_{1} = 1$, or $\gamma_{1} > 1$.
  We write below the parameterized costs $\ALG(B_1)$ for both cases, one can verify that the obtained competitive ratio is below $11/9$ in each case.
  \begin{itemize}
  \item If $\gamma_{1} = 1$, then \ALG\ pays an additional cost of at most $1$ (compared to \OPTstar).
    \ALG\ can avoid paying that additional cost if it already has $x$ at the center, which happens with a probability of at least $1/3^{\epsilon_{1}}$.
    Hence, \ALG\ pays in expectation at most:
    \begin{align*}
      \mathbb{E}[ \ALG(B_{1}) ] \leq 2 \epsilon_{1} + \gamma_{1} + \left(1 - 1/3^{\epsilon_{1}}\right)
    \end{align*}

  \item Otherwise, if $\gamma_{1} > 1$, then \ALG\ pays an additional cost of at most $2$.
    As \ALG\ treats the requests of the block, \ALG's candidate set experiences at most two strict reductions.
    We already covered the case where zero or one reductions occur as \ALG\ treats the requests of the block, we therefore assume that exactly two reductions occur.
    This implies that there exists another \VM\ $y$ that was also kept in the candidate set after the first cheap request (which coincides with the first reduction).
    Just as $x$, $y$ is at the center just before the first cheap request with a probability of at least $1/3^{\epsilon_{1}}$.
    Using the total probability law on \ALG's possible outcomes, it holds:
    \begin{align*}
      \mathbb{E}[\ALG(B_{1})] \leq 2 \epsilon_{1} + \gamma_{1} + 1/3^{\epsilon_{1}} + \frac{3}{2}\cdot \left(1 -  2/3^{\epsilon_{1}}\right) 
    \end{align*}
  \end{itemize}

  \textbf{Case $\bm{q > 1}$.}
  We are going to analyze the expected costs on three different types of blocks: the ambiguous block $B_{1}$, the last simple block $B_{q}$, and the intermediate ambiguous blocks $B_{2}, B_{3}, \dots B_{q-1}$, if they exist.
  We show that the desired competitive ratio holds on each intermediate block, and that it holds on the first and the last blocks combined.

  We first show that the claim holds on any block $B_{k}$ from $B_{2}, \ldots,  B_{q-1}$.
  Let $x$ and $y$ be the two \VMs\ in $\cand$ at the end of $B_{k}$. These \VMs\ are requested in each cheap request of $B_{k}$.
  If $x$ or $y$ is requested in the expensive requests of $B_{k}$ then it stays at the center until the first cheap request with a probability of at least $1/3^{\epsilon_{k}}$.
  Otherwise, $x$ or $y$ was in $\cand$ before $B_{k}$.
  In that case, either $x$ or $y$ were at the center before block $B_{k}$ with uniform probabilities since the previous block is ambiguous.
  As a result, $x$ or $y$ are at the center before the first cheap request with a probability of at least $0.5 \cdot 1/3^{\epsilon_{k}}$.
  Hence, we have
  \begin{align*}
    \ALG(B_{k}) \leq 2 \epsilon_{k} + \gamma_{k} + \left(1 - 2 \cdot \frac{1}{2 \cdot  3^{\epsilon_{k}} }\right)
  \end{align*}
  This expression is the same as in the case $q=1$ and the competitive ratio of $11/9$ holds for the same reason.

  In the following, we show that the desired upper bound holds on blocks $B_{1}$ and $B_{q}$, combined.
  Let us first compute their expected costs separately.

  The costs for $B_{1}$ are computed as follows: let $x$ and $y$ be the two \VMs\ in $\cand$ at the end of $B_{1}$.
  As stated in the case $q=1$, the probability that $x$ (resp. $y$) is at the center just before the first cheap request is at least $1/3^{\epsilon_{1}}$.
  Hence:
  \begin{align*}
    \mathbb{E}[\ALG(B_{1})] \leq 2 \epsilon_{1} + \gamma_{1} + (1 - 2 \cdot 1/3^{\epsilon_{1}})
  \end{align*}

  We next give an upper bound for the expected cost of \ALG\ on $B_{q}$.
  Let $x$ be the final \VM\ in $\cand$ for that block.
  The situation is almost the same as in case $q=1$, except for the probability of $x$ to be at the center before the first cheap request:
  \begin{align*}
    \textrm{if $\gamma_{q} = 1$, then }\\
    \mathbb{E}[\ALG(B_{q}) ] &\leq 2 \epsilon_{q} + \gamma_{q} + \left(1 - \frac{1}{2 \cdot  3^{\epsilon_{q}} }\right),\\
    \textrm{otherwise if $\gamma_{q} > 1$ }\\
    \mathbb{E}[\ALG(B_{q})] &\leq 2 \epsilon_{q} +
                              \gamma_{q} + \frac{1}{2
                              \cdot 3^{\epsilon_{q}}} + \left(1 - 0.5 \cdot \frac{1}{ 2 \cdot 3^{\epsilon_{1}} }\right) \cdot \frac{3}{2}
  \end{align*}
  It is easy to verify that the competitive ratio holds by analyzing the blocks $B_{1}$ and $B_{q}$ combined.

  We now analyze the last group of blocks of the form $A^{*}$.
  As previously shown, the desired competitive ratio holds on all blocks except for the last simple one, the claim therefore holds for that last pattern.

  We partitioned the input sequence into subsets of requests, and proved that the claimed competitive ratio always holds on each subset. Therefore, \ALG\ is $11/9$-competitive.
\end{proof}

\section{Conclusions}
\label{sec:conclusions}

In this paper, we obtained optimal online algorithms for the dynamic graph embedding problem in the star topology.
The optimal competitive ratios for the studied problem turned out to be very low: our randomized online algorithm incurs at most $22\%$ more cost than any algorithm that knows the future requests.

The main direction for future research is to investigate other network topologies beyond the star. We believe that the techniques obtained in this paper can help with this task, and that our algorithms will serve as a fundamental algorithmic building block.

\balance

\bibliographystyle{IEEEtran}

\bibliography{references}

\appendix

\section{Proof of \cref{lem:tracking}}
\label{apx:tracking}
\trackinglemma*

\begin{proof}
  We prove the claim using induction on the request sequence $\sigma$.
  For any request $\sigma_j$, we call $\cand_j$ the candidate set of \ALG\ after it served $\sigma_j$.

  We first prove that the base case for the induction holds, i.e., the claim holds for $\sigma_1$.
  Let $x_{\textrm{init}}$ be the initial central \VM\ for both \ALG\ and \OPTstar, $\cand_{0} = \{x_{\textrm{init}}\}$.
  \ALG\ applies the union behavior on $\sigma_1$ if the initial central \VM\ is not requested. In this case, \OPTstar\ pays $2$ since it has no \VM\ from $\sigma_1$ at the center.
  \ALG\ applies the intersection behavior if the initial central \VM\ is in~$\sigma_1$ and \OPTstar\ pays $1$ (otherwise it would not have a minimum reversed lexicographic order). This concludes the base case.

  For the induction step, let $\sigma_i$ be any request in phase~$P$ with pivot \VM\ $x$.
  By our induction hypothesis, the claim holds for the requests $\sigma_1, \sigma_2 \dots \sigma_{i-1}$.
  We divide the analysis into the following two steps. 
  We first prove that if \OPTstar\ pays $1$ for $\sigma_i$ then \ALG\ applies its intersection behavior. Afterward, we prove that if \OPTstar\ pays $2$ then \ALG\ applies its union behavior.

  \textbf{\OPTstar\ pays $\bm{1}$} for $\sigma_i$, i.e., $\sigma_i$ is a cheap request.
  It follows that $\sigma_i$ contains the pivot \VM\ $x$.
  We show that \ALG\ applies its intersection behavior on $\sigma_{i}$.

  We claim that $x$ was in $\cand$ at the beginning of phase $P$.
  If $P$ is the first phase of the input sequence, then $x$ was inserted in $\cand$ when \ALG\ was initialized.
  Otherwise, if $P$ is not the first phase, \cref{lem:phaseStartsWithX} states that the phase $P$ starts with an expensive request that contains $x$. According to the induction hypothesis, \ALG\ applied its union behavior for it, meaning that $x$ was inserted in $\cand$.

  We next claim that $x$ stays in $\cand$ from the beginning of the phase until $\sigma_i$.
  By the induction hypothesis, an expensive request only adds new \VMs\ to $\cand$ whereas a cheap request intersects $\cand$ with its requested \VMs, including $x$. $x$ therefore stays in $\cand$.

  The pivot is therefore contained in $\cand$ after $\sigma_{i-1}$ is served.
  Since $\sigma_i$ contains the pivot, \ALG\ applies its intersection behavior for $\sigma_i$.

  \textbf{\OPTstar\ pays $\bm{2}$} for $\sigma_i$, i.e., $\sigma_i$ is an expensive request.
  We want to prove that \ALG\ applies its union behavior.
  By means of contradiction, assume that \ALG\ applies its intersection behavior for request $\sigma_i$.
  We now construct a new optimal offline algorithm \OFF\ that will either pay a lower cost than \OPTstar\ or have a lower lexicographic order.

  The description of algorithm \OFF\ is as follows.
  Since \ALG\ applies its intersection behavior on $\sigma_i$, it means that $\cand_{i-1} \cap \sigma_i$ is not empty, let $y \in \cand_{i-1} \cap \sigma_i$. 
  Hence, $y \in \cand_{i-1}$. Let $\sigma_{k(y)}$ be the earliest request such that $\forall j \in [k(y), i-1]$, $y \in \cand_j$.
  Let \OFF\ act like \OPTstar\ up until $k(y)$.
  Before serving $\sigma_{k(y)}$, let \OFF\ migrate $y$ to the center and keep it there until it serves $\sigma_i$. There are now two possibilities.
  \OFF\ either keeps $y$ at the center until the end of $P$ if there is no more occurrences of $x$ in $P$, or it migrates $x$ to the center at the next occurrence of $x$. In both cases, \OFF\ acts just like \OPTstar\ on the remaining requests.

  Now that we described the new algorithm \OFF, we compare it to \OPTstar\ and derive the contradiction.
  We first claim that \OFF\ pays no more than \OPTstar\ minus $1$ from the beginning of the input sequence until it serves~$\sigma_i$.
  If~$k(y)=0$, then $y$ is at the center from the very beginning and \OFF\ does not need to migrate.
  Otherwise, if $k(y) \geq 1$ then, by the minimality of $k(y)$, it means that $y$ was inserted in $\cand$ right before $\sigma_{k(y)}$ is served.
  We deduce that \ALG\ applied its union behavior for request~$\sigma_{k(y)}$ and that~$\sigma_{k(y)}$ contains $y$. Note that this is the only possibility for \ALG\ to add a new \VM\ to $\cand$ once $\cand$ was initialized.
  Using the induction hypothesis, we conclude that \OPTstar\ paid $2$ for the request $\sigma_{k(y)}$: \OFF\ migrates $y$ to the center for the same cost as \OPTstar.

  After \OFF\ has $y$ at the center, we show that \OFF\ pays the same cost as \OPTstar\ minus $1$ until it serves $\sigma_i$.
  Since \ALG\ evicts any \VM\ from $\cand$ that is not requested when it treats one of the cheap requests, any cheap request from $\sigma_{k(y)}$ to $\sigma_{i-1}$ (or from $\sigma_{1}$ to $\sigma_{i-1}$ if $k(y) = 0$) must contain $y$.
  Otherwise, $y$ would have been evicted from $\cand$ which would contradict the definition of index $k(y)$.
  \OFF\ therefore pays at most the same cost as \OPTstar\ until it serves $\sigma_{i-1}$. 
  Finally, \OFF\ pays one less than \OPTstar\ on $\sigma_i$ since $y \in \sigma_i$ and $x \notin \sigma_i$.
  Thus the intermediate claim follows.

  We now compare \OFF\ and \OPTstar\ after they served $\sigma_i$ and exhibit a contradiction in all cases.
  We partition our analysis into two cases, depending on the two possible behaviors of \OFF\ after $\sigma_i$.
  Assume that $x$ is not requested in $P$ after $\sigma_i$.
  \OPTstar\ pays $2$ for each of those requests while \OFF\ pays no more than $2$.
  Since \OFF\ already paid a strictly lower cost than \OPTstar, \OFF\ globally pays a lower cost than \OPTstar. This is a contradiction.

  Otherwise, assume $x$ is requested after $\sigma_i$ in $P$, let $\sigma_{k(x)}$ be the first of those requests.
  \OFF\ keeps $y$ at the center and migrates $x$ to the center right before serving $\sigma_{k(x)}$.
  Thus \OFF\ pays the same cost as \OPTstar\ since it spares $1$ on $\sigma_{i}$ and pays $1$ more on $\sigma_{k(x)}$. 
  Compared to \OPTstar, the time that $x$ spends at the center was strictly reduced.
  We created an optimal algorithm with a lower reversed lexicographic order, which is a contradiction.

  We showed that \ALG\ does not apply its intersection behavior on an expensive request, as it leads to a contradiction in all cases. We conclude that \ALG\ always applies its union behavior for expensive requests.
  This claim holds for all $\sigma_{i}$ and therefore it holds for any request in  $\sigma$.
\end{proof}

\end{document}